\documentclass{amsart}

\usepackage[english]{babel}
\usepackage[T1]{fontenc}
\usepackage{mathrsfs}

\makeatletter
\def\@settitle{\begin{center}\baselineskip14\p@\relax
  \bfseries
  \uppercasenonmath\@title
  \@title
  \\\mdseries\footnotesize *~~*~~*
  \ifx\@subtitle\@empty\else
     \\[-0.6ex]\uppercasenonmath\@subtitle
     \footnotesize\mdseries\@subtitle
  \fi
  \end{center}}
\def\subtitle#1{\gdef\@subtitle{#1}}
\def\@subtitle{}
\makeatother

\usepackage{paralist}
\usepackage{enumitem}
\usepackage{caption}
\usepackage{wrapfig}
\usepackage{subfig}

\usepackage[numbers,longnamesfirst]{natbib}
\usepackage{url}
\usepackage{hyperref}
\hypersetup{
  bookmarks=true,
  colorlinks = true,
  linkcolor  = duckBlue,          citecolor  = goldOrange,          filecolor  = duckBlue,          urlcolor   = goldOrange
}
\usepackage[capitalize,nameinlink,noabbrev]{cleveref}

\usepackage{amsmath,amsfonts,amssymb}
\usepackage{upgreek}
\usepackage{nicefrac}
\usepackage{faktor}
\usepackage{xspace}

\usepackage{bbm}
\usepackage{multirow}
\usepackage{graphicx}
\usepackage{tabularx}
\usepackage{booktabs}

\usepackage[usenames,dvipsnames]{xcolor}
\definecolor{duckBlue}{HTML}{005A66}
\colorlet{goldOrange}{orange!80!black}

\usepackage{tikz}
\usetikzlibrary{matrix}
\usetikzlibrary{patterns}

\usepackage{tikz-cd}

\usepackage[linesnumbered, longend]{algorithm2e}
\makeatletter
\newcommand{\algorithmstyle}[1]{\renewcommand{\algocf@style}{#1}}
\makeatother

\DontPrintSemicolon

\SetKwSty{kwfont}
\SetKwInOut{Parameter}{Parameters}
\SetKw{And}{and}
\SetKw{Goto}{goto}
\SetKw{True}{true}
\SetKw{Break}{break}
\SetKw{False}{false}
\SetKw{Downto}{downto}
\usepackage[most]{tcolorbox}
\tcbset{
    algotitle/.style={title={\strut
        Algorithm~\thetcbcounter\ifstrempty{#1}{\ignorespaces}{:~#1}}}}
\newtcolorbox[auto counter]{algo}[1][]{colback=white,
    colframe=duckBlue,
    boxrule=1pt,
    titlerule=0pt,
    sharp corners=all,
    colbacktitle=white,enhanced,
    attach boxed title to top center={yshift=-10pt},
    boxed title style={boxrule=-1pt},
    fonttitle=\sffamily,
    coltitle=duckBlue,
    algotitle={},
    #1
}

\usepackage{todonotes}

\usepackage{amsthm}
\usepackage{mathtools}
\theoremstyle{theorem}
\newtheorem{heuristic}{Heuristic}[section]
\newtheorem{definition}{Definition}[section]
\newtheorem{theorem}{Theorem}[section]
\newtheorem{lemma}{Lemma}[section]
\newtheorem{proposition}{Proposition}[section]
\newtheorem{corollary}{Corollary}[section]

\newtheorem*{remark}{Remark}
\newtheorem*{exemple}{Exemple}

\let\originalleft\left
\let\originalright\right
\renewcommand{\left}{\mathopen{}\mathclose\bgroup\originalleft}
\renewcommand{\right}{\aftergroup\egroup\originalright}

\newcommand{\LLL}{\mbox{\textsc{lll}}}
\newcommand{\SVP}{\mbox{\textsc{svp}}}
\newcommand{\HSVP}{\mbox{\textsc{hsvp}}}

\newcommand{\DBKZ}{\mbox{\textsc{dbkz}}}

\DeclareMathOperator{\rk}{rk}

\DeclareMathOperator{\covol}{covol}
\DeclareMathOperator*{\argmin}{argmin}

\newcommand{\dual}[1]{{#1}^{\vee}}
\newcommand{\orth}[1]{#1_{\bot}}

\newcommand{\vect}[1]{{#1}}

\newcommand{\inner}[2]{\langle \vect{#1},\vect{#2} \rangle}

\newcommand{\ZZ}{\mathbf{Z}}
\newcommand{\RR}{\mathbf{R}}

\newcommand{\QQ}{\mathbf{Q}}

\newcommand{\Lat}{\Lambda}

\newcommand{\NP}{\mbox{\textbf{NP}\xspace}}

\newcommand{\CVP}{\mbox{\textsc{cvp}\xspace}}
\newcommand{\ACVP}{\mbox{\textsc{approx-cvp}\xspace}}
\newcommand{\ASVP}{\mbox{\textsc{approx-svp}\xspace}}
\newcommand{\ACVPP}{\mbox{\textsc{approx-cvpp}\xspace}}

\newcommand{\algName}[1]{{\color{duckBlue}\textbf{\textsf{#1}}\color{black}}}

\newcommand{\bigO}[1]{\mathrm O\left(#1\right)}
\newcommand{\littleO}[1]{\mathrm{o}\left(#1\right)}

\newcommand{\oracle}{$\mathcal{O}$}
\newcommand{\SL}{\mathrm{SL}}
\newcommand{\dist}{\mathrm{dist}}
\newcommand{\Span}{\mathrm{span}}
 
\begin{document}
\clearpage{}\newcommand{\plane}[1]{
	(-1.95, #1, 1.35) --
	++(3.6, 0.6, 0.0) --
	++(0.3, -1.8, -2.7) --
	++(-3.6, -0.6, -0.0) --
	cycle
}

\newcommand{\nearestplane}{
\draw[dotted] (0,-1.3,0) -- (0,-0.6,0);
	\draw[->] (0, 1.0,0) -- (0,2,0);
	\draw[-, dotted] (0, 0.28,0) -- (0,1.0,0);
	\draw[-,dotted] (0, -0.6,0) -- (0,0.2,0);
  \node[above]   (0) at (0,2,0) {\footnotesize ${\Lambda}/{\Lambda_2}$};
	\draw[goldOrange,thick,  ->] (0,-0.6,0) -- (1.8,0.68,0);
  \draw[goldOrange, fill] (1.8,0.68,0) circle  [radius=.5pt];
  \node[right]  (3) at (1.85,0.55,0) {\color{goldOrange}\footnotesize$t$};
	\draw[fill=gray!20, opacity=0.3]\plane{-0.8};
  \draw[fill] (0,-1.3,0) circle  [radius=.5pt];
	\draw[fill=gray!30!gray, opacity=0.6]\plane{0};
  \node[right] (0) at (2.45, 0.95) {\color{duckBlue}\footnotesize$v+\Lambda_2$};
  \draw[fill] (0,-0.6,0) circle  [radius=.5pt];
  \node[right]   (0) at (2.45, -0.65) {\footnotesize$\Lambda_2$};
	\draw[fill=gray!20, opacity=0.3]\plane{0.8};
	\draw[fill=orange, opacity=0.3]\plane{1.38};
  \node[right] (0) at (2.45, 0.71) {\color{goldOrange}\footnotesize$t+\Lambda_2$};
	\draw[fill=duckBlue, opacity=0.3]\plane{1.6};
  \draw[fill] (0,0.2,0) circle  [radius=.5pt];
	\draw[-] (0, -1.3,0) -- (0,-1.1,0);
	\draw[-] (0, 0.2,0) -- (0,0.28,0);
	\draw[-] (0, -0.6,0) -- (0,-0.3,0);

\draw[fill=duckBlue, color=duckBlue] (0,1.0,0) circle  [radius=.5pt];
  \draw[fill=goldOrange, color=goldOrange] (0,0.78,0) circle  [radius=.5pt];

  \node[right]  (1) at (-0.25,-0.6) {\footnotesize 0};
  \node[right]  (4) at (-0.59,1.0,0) {\color{duckBlue}\footnotesize$\pi_2(v)$};
  \node[right]  (5) at (-0.59,0.78,0)
  {\color{goldOrange}\footnotesize$\pi_2(t)$};
}

\newcommand{\nearestcospace}{
\draw[dotted] (0,-1.3,0) -- (0,-0.6,0);
  \draw[] (0,-1.3,0) -- (0,-1.09,0);
	\draw[color=duckBlue, thick, dotted] (1,-1.3,0) -- (1,-0.6,0);
	\draw[color=goldOrange, dotted, thick] (1.2,-1.3,0) -- (1.2,-0.6,0);
  \draw[fill=duckBlue, color=duckBlue] (1,-0.6,0) circle  [radius=.5pt];
  \draw[fill=goldOrange, color=goldOrange] (1.2,-0.6,0) circle  [radius=.5pt];
	\draw[color=duckBlue, thick] (1,-1.3,0) -- (1,-0.9,0);
	\draw[color=goldOrange, thick] (1.2,-1.3,0) -- (1.2,-0.88,0);
	\draw[gray,dotted] (-1,-1.3,0) -- (-1,-0.6,0);
  \draw[fill] (-1,-0.6,0) circle  [radius=.5pt];
	\draw[gray,] (-1,-1.3,0) -- (-1,-1.27,0);
	\draw[gray,dotted] (0,-1.3,1) -- (0,-0.6,1);
  \draw[fill] (0,-0.6,1) circle  [radius=.5pt];
	\draw[gray,] (0,-1.3,1) -- (0,-0.78,1);
	\draw[gray,dotted] (0,-1.3,-1) -- (0,-0.6,-1);
  \draw[fill] (0,-0.6,-1) circle  [radius=.5pt];
	\draw[gray,dotted] (-1,-1.3,1) -- (-1,-0.6,1);
  \draw[fill] (1,-0.6,1) circle  [radius=.5pt];
	\draw[gray,] (-1,-1.3,1) -- (-1,-0.78,1);
	\draw[gray,dotted] (-1,-1.3,-1) -- (-1,-0.6,-1);
  \draw[fill] (-1,-0.6,-1) circle  [radius=.5pt];
  \node[right]  (3) at (-0.15,0.4,0) {\footnotesize$\Lat_1$};
  \draw[fill] (-1,-0.6,1) circle  [radius=.5pt];

	\draw[fill=white!70!gray, opacity=0.6]\plane{0};
	\draw[->] (0,-0.6,0) -- (1.2,0.68,0);

	\draw[] (0,-0.6,0) -- (0,0.21,0);

	\draw[color=duckBlue, thick] (1,-0.6,0) -- (1,1,0);
	\draw[color=goldOrange, thick] (1.2,-0.6,0) -- (1.2,1,0);

	\draw[gray,] (-1,-0.6,0) -- (-1,0.05,0);
	\draw[gray,] (0,-0.6,1) -- (0,0.53,1);
	\draw[gray,] (0,-0.6,-1) -- (0,-0.11,-1);
	\draw[gray,] (-1,-0.6,1) -- (-1,0.55,1);
	\draw[gray,] (-1,-0.6,-1) -- (-1,-0.28,-1);
  \draw[gray,] (1,-1.3,1) -- (1,0.7,1);

  \node[right]  (3) at (1.2,0.65,0) {\color{goldOrange}\footnotesize$t$};
  \node[right]  (3) at (1.19,-0.6,0)
  {\color{goldOrange}\tiny$\pi_1(t)$};
  \node[right]  (3) at (1.1,1.1,0) {\color{goldOrange}\footnotesize$t+\Lambda_{1}$};
  \node[right]  (3) at (0.5,1.1,0) {\color{duckBlue}\footnotesize$v+\Lambda_{1}$};
  \node[right]  (3) at (2.5,-0.68,0)
  {\footnotesize$\faktor{\Lat}{\Lambda_{1}}$};
  \node[right]  (3) at (0.45,-0.6, 0) {\color{duckBlue}\tiny$\pi_1(v)$};
  \node[right]  (1) at (-0.25,-0.6) {\footnotesize 0};
\draw[fill] (0,-0.6,0) circle  [radius=.5pt];
  \draw[fill=goldOrange,color=goldOrange] (1.2,0.68,0) circle  [radius=.5pt];
}

\clearpage{}

\title{The nearest-colattice algorithm}
\subtitle{\emph{Time-approximation tradeoff for approx-CVP}}
\author{Thomas Espitau${}^\divideontimes$ \and Paul Kirchner${}^\star$}
\address{${}^\divideontimes$ NTT Corporation, Toky\=o, Japan\\ ${}^\star$
Rennes University, Rennes, France
}
\email{t.espitau@gmail.com, paul.kirchenr@irisa.fr}

\begin{abstract}
  In this work, we exhibit a hierarchy of polynomial time algorithms solving
  approximate variants of the Closest Vector Problem (\CVP).  Our first
  contribution is a heuristic algorithm achieving the same distance tradeoff
  as \HSVP{} algorithms, namely $\approx
  \beta^{\frac{n}{2\beta}}\covol(\Lat)^{\frac{1}{n}}$ for a random lattice
  $\Lat$ of rank $n$.  Compared to the so-called Kannan's embedding technique,
  our algorithm allows using precomputations and can be used for efficient
  batch \CVP~instances.  This implies that some attacks on lattice-based
  signatures lead to very cheap forgeries, after a precomputation.  Our second
  contribution is a \emph{proven} reduction from approximating the closest
  vector with a factor $\approx n^{\frac32}\beta^{\frac{3n}{2\beta}}$ to the
  Shortest Vector Problem (\SVP) in dimension $\beta$.
\end{abstract}

\maketitle

\section{Introduction}
\label{sec:introduction}

\subsection*{Lattices, CVP, SVP}
In a general setting, a real \emph{lattice} $\Lat$ is a finitely
generated free $\ZZ$-module, endowed with a positive-definite quadratic
form on its ambient space $\Lat\otimes_\ZZ\RR$, or equivalently is a discrete
subgroup of a Euclidean space.

A fundamental lattice problem is the \emph{Closest Vector Problem}, or
\CVP~for short. The goal of this problem is to find a
lattice point that is closest to a given point in its ambient space.
This problem is provably difficult to solve, being actually a
\NP-hard problem.
It is known to be harder than the \emph{Shortest Vector Problem}
(\SVP)~\cite{goldreich1999approximating}, which asks for the
shortest non-zero lattice point. \SVP{} is, in turn, the cornerstone of
lattice reduction algorithms (see for
instance~\cite{Schnorr87,C:HanPujSte11,EC:MicWal16}). These algorithms are at
the heart of lattice-based cryptography~\cite{regev2009lattices}, and are
invaluable in plenty of computational problems, including Diophantine
approximation, algebraic number theory or optimization
(see~\cite{nguyen2010lll} for a survey on the applications of the \LLL{}
algorithm).

\subsection*{On CVP-solving algorithms}
There are three families of algorithms solving \CVP:
\begin{description}
  \item[Enumeration algorithms] consisting in recursively explore all vectors
    in a set containing a closest vector. Kannan's algorithm takes time
    $n^{\bigO{n}}$ and polynomial space~\cite{kannan1987minkowski}. This
    estimate was later refined to $n^{\frac{n}{2}+\littleO{n}}$ by Hanrot and
    Stehl\'e~\cite{C:HanSte07}.

  \item[Voronoi cell computation] Micciancio and Voulgaris' Voronoi cell
    algorithm solves \CVP~in time $(4+\littleO{1})^n$ but uses a space of
    $(2+\littleO{1})^n$~\cite{SODA:MicVou10}.

  \item[Sieving algorithms] where
    vectors are combined in order to get closer and closer to the target
    vector. Heuristic variants take time as low as
    $(4/3+\littleO{1})^{\frac{n}{2}}$~\cite{SODA:BDGL16},
    but proven variants of classical
    sieves~\cite{ajtai2002sampling,blomer2009sampling,eisenbrand2011covering}
    could only solve \CVP{} with approximation factor $1+\epsilon$ at a cost
    in the exponent. In 2015, a $(2+\littleO{1})^n$ sieve for \emph{exact}
    \CVP{} was finally proven by Aggarwal, Dadush and
    Stephen-Davidowitz~\cite{FOCS:AggDadSte15} thanks to the properties of
    discrete Gaussians.
\end{description}

Many algorithms for solving its relaxed variant, \ACVP{}, have been
proposed. However, they come with caveats.  For example, Dadush, Regev and
Stephens-Davidowitz~\cite{dadush2014closest} give algorithms for this problem,
but only with exponential time precomputations.  Babai~\cite[Theorem
3.1]{Babai86} showed that one can reach an $2^{\frac{n}{2}}$-approximation
factor for \CVP{} in polynomial time.  To the authors'
knowledge,
this has never been improved (while keeping the polynomial-time requirement),
though the approximation factor for \SVP{} has been significantly
reduced~\cite{Schnorr87,C:HanPujSte11,EC:MicWal16}.

We aim at solving the relaxed version of
\CVP{} for relatively large approximation factors, and study the tradeoff
between the quality of the approximation of the solution found and the time
required to actually find it.
In particular, we exhibit a hierarchy of polynomial-time algorithms
solving \ACVP, ranging from Babai's nearest plane algorithm to
an actual \CVP~oracle.

\subsection*{Contributions and summary of the techniques}
In~\cref{sec:nearest_coplane} we introduce our so-called
\algName{Nearest-Colattice} algorithm. Inspired by Babai's algorithm, it shows
that in practice, we can achieve the performance of Kannan's embedding but
with a basis which is \emph{independent} of the target vector.
Denote by $T(\beta)$ (resp.
$T_{\CVP}(\beta)$) the time required to solve $\sqrt{\beta}$-Hermite-\SVP{}
(resp. exactly solve \CVP{}) in rank $\beta$).
Quantitatively, we show that:
\begin{theorem}[Informal]
  Let $\beta>0$ be a positive integer and $B$ be a basis of a lattice $\Lat$
  of rank $n>2\beta$. After precomputations using a time bounded by
  $T(\beta)(n+\log \|B\|)^{\bigO{1}}$, given a target $t\in\Lat_\RR$
  and under a heuristic on the covering radius of random lattice, the
  algorithm \algName{Nearest-Colattice} finds a vector $x\in \Lat$ such that
  \[ \|x-t\| \leq \Theta(\beta)^{\frac{n}{2\beta}}\covol(\Lat)^{\frac{1}{n}}
  \] in time $T_{\CVP}(\beta)(n+\log \|t\|+\log \|B\|)^{\bigO{1}}$.
\end{theorem}
Furthermore, the structure of the algorithms allow time-memory tradeoff and
batch \CVP{} oracle to be used.

We believe that this algorithm has been in the folklore for some time, and
it is somehow hinted in \texttt{ModFalcon}'s security
analysis~\cite[Subsection 4.2]{modfalcon}, but without analysis of
the heuristics introduced.

Our second contribution is an \ACVP{} algorithm, which gives a time-quality
tradeoff similar to the one given by the \textsc{bkz}
algorithm~\cite{Schnorr87,C:HanSte07}, or variants of
it~\cite{STOC:GamNgu08,aggarwal2019slide}.
Note however that the approximation factor is significantly higher than
the corresponding theorems for \ASVP. Written as a reduction, we prove that, for
a $\gamma$-\HSVP{} oracle \oracle{}:

\begin{theorem}[\ACVPP{} oracle from \ASVP{} oracle]
  Let $\Lat$ be a lattice of rank $n$. Then one can solve the
  $(n^\frac32\gamma^3)$-closest vector problem in $\Lat$, using $2n^2$ calls
  to the oracle \oracle{} during precomputation, and polynomial-time
  computations.
\end{theorem}
Babai's algorithm requires that the Gram-Schmidt norms do not decrease by too much in
the reduced basis. While this is true for a \LLL{} reduced basis~\cite{LLL82}, we
do not know a way to guarantee this in the general case. To overcome this
difficulty,
the proof technique goes as follows:
first we show that it is possible to find a vector within distance
$\frac{\sqrt{n}\gamma}{2}\lambda_n(\Lat)$
of the target vector, with the help of a highly-reduced basis.
This is not enough, as the target can be very closed compared to
$\lambda_n(\Lat)$.  We treat this peculiar case by finding a short vector
in the dual lattice and then directly compute the inner product of the close
vectors with our short dual vector.  In the other case, Banaszczyk's
transference theorem~\cite{banaszczyk1993new} guarantees that
$\lambda_n(\Lat)$ is comparable to the distance to the
lattice, so that we can use our first algorithm directly.

\begin{remark}
  Based on a result due to Kannan (see for
  instance~\cite{dubey2011approximating}) that $\sqrt{n}\gamma^2$ \CVP{}
  reduces to $\gamma$-\SVP{}. Combined with the reduction from
  $\gamma^2$-\SVP{} to $\gamma$-\HSVP{} of \cite{lovasz1986algorithmic},
  we get a polynomial time reduction from $\sqrt{n}\gamma^4$-\CVP{} to
  $\gamma$-\HSVP.  Hence, our result is better when $n^{\frac32}\gamma^3$ is
  smaller than $\sqrt{n} \gamma^4$, i.e., when $n < \gamma$.
\end{remark}

\section{Algebraic and computational background}
\label{sec:background}

In this preliminary section, we recall the notions of geometry of numbers used
throughout this paper, the computational problems related to $\SVP$ and
$\CVP$, and a brief presentation of some lattice reduction algorithms
solving these problems.

\subsection*{Notations and conventions}
\subsubsection*{General notations}
The bold capitals $\ZZ$, $\QQ$ and $\RR$ refer as usual to the
ring of integers and respectively the field of rational and real numbers.
Given a real number $x$, the integral roundings
\emph{floor}, \emph{ceil} and \emph{round to the nearest integer} are
denoted respectively by $\lfloor x\rfloor, \lceil x\rceil, \lfloor x
\rceil$. All logarithms are taken in base $2$, unless explicitly
stated otherwise.

\subsubsection*{Computational setting}
The generic complexity model used in this work is the random-access
machine (RAM) model and the computational cost is measured in
operations.

\subsection{Euclidean lattices and their geometric invariants}
\label{sec:euclidean_spaces}

\subsubsection{Lattices}
\label{sec:definitions}

\begin{definition}[Lattice]
  A (real) \emph{lattice}  $\Lat$ is a finitely generated free
  $\ZZ$-module, endowed with a Euclidean norm $\|.\|$ on the real
  vector space $\Lat_\RR = \Lat\otimes_\ZZ \RR$.
\end{definition}

We may omit to write down the norm to refer to a lattice $\Lat$ when any
ambiguity is removed by the context. By definition of a
finitely-generated free module, there exists a finite family $(v_1,
\ldots, v_n) \in \Lat^n$ such that $\Lat = \bigoplus_{i=1}^n v_i \ZZ$,
called a \emph{basis} of $\Lat$. Every basis has the same number of
elements $\rk(\Lat)$, called the rank of the lattice.

\subsubsection{Sublattices, quotient lattice}
\label{sec:sublattices_quotients}

Let $(\Lat, \|\cdot\|)$ be a lattice, and let $\Lat'$ be a submodule of
$\Lat$. Then the restriction of $\|\cdot\|$ to $\Lat'$ endows $\Lat$ with a
lattice structure. The pair $(\Lat', \|\cdot\|)$ is called a \emph{sublattice}
of $\Lat$. In the following of this paper, we restrict ourselves to so-called
\emph{pure sublattices}, that is such that the quotient
$\faktor{\Lat}{\Lat'}$ is torsion-free. In this case, the quotient can be
endowed with a canonical lattice structure by defining: \[
  \|v+\Lat'\|_{\Lat/\Lat'} = \inf_{v'\in\Lat'_\RR} \|v-v'\|_\Lat.
\]
This lattice is isometric to the projection of $\Lat$ orthogonally to
the subspace of $\Lat_\RR$ spanned by $\Lat'$.

\subsubsection{On effective lifting.} Given a coset $v+\Lat'$ of the quotient
$\faktor{\Lat}{\Lat'}$, we might need
to find a representative of this class in $\Lat$. While any element could be
theoretically taken, from an algorithmic point of view, we shall take an
element of norm somewhat small, so that its coefficients remain polynomial in
the input representation of the lattice. An effective solution to do so
consists in using for instance the \emph{Babai's rounding} or \emph{Babai's
  nearest
plane} algorithms. For completeness purpose we recast here the pseudo-code of
such a \algName{Lift} function using the nearest-plane procedure.

\begin{algo}[algotitle={Lift (by Babai's nearest plane)},
  label=alg:lift]
  \sffamily
  \begin{algorithm}[H]
\BlankLine
    \KwIn{A lattice basis $B = (v_1, \ldots, v_k)$ of $\Lat'$ in $\Lat$, a vector
    $t\in\Lat_\RR$.}
    \KwResult{A vector of the class $\tilde{t}+\Lat'\in \Lat$.}
\BlankLine
    Compute the Gram-Schmidt orthogonalization $(v_1^*, \ldots, v_k^*)$ of
    $B$\;
    $s \gets -t$\;
    \For{$i=k$ \Downto $1$}{$s \gets s-\left\lfloor\frac{\inner{s}{v_i^*}}{\|v_i^*\|^2}
    \right\rceil v_i$\;
  }
  \Return $t+s$
\end{algorithm}
\end{algo}

\subsubsection{Orthogonality and algebraic duality}
\label{sec:algebraic_nuality}

The \emph{dual} lattice $\dual{\Lat}$ of a lattice $\Lat$
is defined as the module $\textrm{Hom}(\Lat,\ZZ)$ of integral linear forms,
endowed with the derived norm defined by
\[\|\varphi\| = \inf_{v\in\Lat_\RR\setminus\{0\}}
\frac{|\varphi(v)|}{\|v\|_\Lat}\]
for $\varphi\in\dual{\Lat}$. By Riesz's representation theorem, it is
isometric to:
\[\{ x\in\Lat_\RR ~|~ \inner{x}{v}\in\ZZ,
\forall v\in\Lat\}\] endowed with the dual of $\|\cdot\|_\Lat$.

Let $\Lat'\subset\Lat$ be a sublattice. Define its
\emph{orthogonal} in $\Lat$ to be the sublattice
\[\orth{\Lat'} = \{x \in \dual{\Lat} : \inner{x}{\Lat'} = 0\}\] of
$\dual{\Lat}$. It is isometric to
$\dual{\left(\faktor{\Lat}{\Lat'}\right)}$, and by biduality
$\orth{\dual{\Lat'}}$ shall be identified with $\faktor{\Lat}{\Lat'}$.

\subsubsection{Filtrations}
\label{sec:filtrations}
A filtration (or flag) of a lattice $\Lat$ is an increasing sequence of
submodules of $\Lat$, i.e.\ each submodule is a proper
submodule of the next:
$\{0\}=\Lat_{0}\subset \Lat_{1}\subset \Lat_{2}\subset \cdots \subset
\Lat_{k}=\Lat.$
If we write the $\rk(\Lat_i) = d_i$, then we have:
$ 0=d_{0}<d_{1}<d_{2}<\cdots <d_{k}=\rk(\Lat),$
A filtration is called \emph{complete} if $d_i = i$ for all i.

\subsubsection{Successive minima, covering radius and transference}
\label{sec:minima}

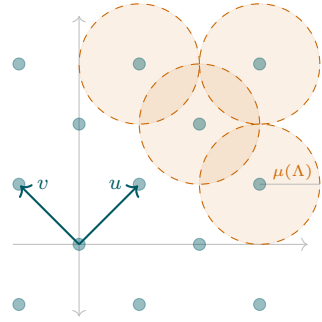
\begin{wrapfigure}{R}{0.3\textwidth}
  \centering
  \begin{tikzpicture}[scale=0.8]
\draw [lightgray] [<->] (0,-1.2) -- (0,3.8);
    \draw [lightgray] [->] (-1.1,0) -- (3.8,0);
\draw [duckBlue, thick,shorten >=0.7 ] [->] (0,0) -- (1,1);
    \draw [duckBlue, thick, shorten >= 0.7] [->] (0,0) -- (-1,1);
    \draw [duckBlue] (0.6,1) node {\footnotesize$u$};
    \draw [duckBlue] (-0.6,1) node {\footnotesize $v$};
\draw [goldOrange, dashed] (2,2) circle [radius=1];
    \draw [goldOrange, dashed] (3,1) circle [radius=1];
    \draw [goldOrange, dashed] (1,3) circle [radius=1];
    \draw [goldOrange, dashed] (3,3) circle [radius=1];
    \draw [goldOrange, fill, opacity=0.1] (2,2) circle [radius=1];
    \draw [goldOrange, fill, opacity=0.1] (3,1) circle [radius=1];
    \draw [goldOrange, fill, opacity=0.1] (1,3) circle [radius=1];
    \draw [goldOrange, fill, opacity=0.1] (3,3) circle [radius=1];
    \draw [lightgray] [-] (3,1) -- (4,1);
    \draw [goldOrange] (3.55,1.2) node {\tiny $\mu(\Lat)$};
\draw [fill,duckBlue,opacity=.4] (-1,3) circle [radius=0.1];
    \draw [fill,duckBlue,opacity=.4] (1,-1) circle [radius=0.1];
    \draw [fill,duckBlue,opacity=.4] (-1,-1) circle [radius=0.1];
    \draw [fill,duckBlue,opacity=.4] (-1,1) circle [radius=0.1];
    \draw [fill,duckBlue,opacity=.4] (3,-1) circle [radius=0.1];
    \draw [fill,duckBlue,opacity=.4] (0,0) circle [radius=0.1];
    \draw [fill,duckBlue,opacity=.4] (0,2) circle [radius=0.1];
    \draw [fill,duckBlue,opacity=.4] (1,1) circle [radius=0.1];
    \draw [fill,duckBlue,opacity=.4] (1,3) circle [radius=0.1];
    \draw [fill,duckBlue,opacity=.4] (3,1) circle [radius=0.1];
    \draw [fill,duckBlue,opacity=.4] (2,0) circle [radius=0.1];
    \draw [fill,duckBlue,opacity=.4] (2,2) circle [radius=0.1];
    \draw [fill,duckBlue,opacity=.4] (3,3) circle [radius=0.1];
  \end{tikzpicture}
  \caption{\label{fig:covering}Covering radius
    $\color{goldOrange}\mu(\Lambda)$ of a two
  dimensional lattice $\Lat$.}
\end{wrapfigure}
Let $\Lat$ be a lattice of rank $n$. By discreteness in $\Lat_\RR$, there
exists a vector of minimal norm in $\Lat$. This parameter is called the
\emph{first minimum} of the lattice and is denoted by $\lambda_1(\Lat)$.  An
equivalent way to define this invariant is to see it as the smallest positive real
$r$ such that the lattice points inside a ball of radius $r$ span a space of
dimension 1. This definition leads to the following generalization, known as
successive minima.
\begin{definition}[Successive minima]
  Let $\Lat$ be a lattice of rank $n$. For $1\leq i\leq n$, define the $i$-th
  minimum of $\Lat$ as
  \[\hspace{-12em}\lambda_i(\Lat) = \inf \{r\in\RR |
  \dim(\Span(\Lat \cap B(0, r))) \geq i\}.\]
\end{definition}

\begin{definition}
  The covering radius a lattice $\Lat$ or rank $n$ is defined as
  \[\mu(\Lat) = \max_{x\in \Lat_\RR} \dist(x,\Lat).\]
\end{definition}
It means that for any vector of the
ambient space $x\in \Lat_\RR$ there exists a lattice point $v \in \Lat$
at distance smaller than $\mu(\Lat)$.

We now recall Banaszczyk's transference theorem,
relating the extremal minima of a lattice and its dual:

\begin{theorem}[Banaszczyk's transference theorem~\cite{banaszczyk1993new}]
  \label{thm:transference}
  For any lattice $\Lambda$ of dimension $n$, we have \[\hspace{-12em} 1\leq
  2\lambda_1(\dual{\Lambda})\mu(\Lambda)\leq n,\]
  implying,
  \[\hspace{-12em}1\leq \lambda_1(\dual{\Lambda})\lambda_n(\Lambda)\leq n. \]
\end{theorem}

\subsection{Computational problems in geometry of numbers}
\subsubsection{The shortest vector problem}
\label{sec:shortest_vector_problem}
In this section, we introduce formally the \SVP~problem and its
variants and discuss their computational hardness.

\begin{definition}[$\gamma$-\SVP]
  Let $\gamma = \gamma(n) \geq 1$. The $\gamma$-Shortest Vector
  Problem ($\gamma$-\SVP) is defined as follows.

  \begin{description}
    \item[Input] A basis $(v_1, \ldots, v_n)$ of a lattice $\Lat$ and
      a target vector $t \in\Lat_\RR$.
    \item[Output] A lattice vector $v\in\Lat\setminus\{0\}$ satisfying
      $\|v\|\leq \gamma \lambda_1(\Lat)$.
  \end{description}
\end{definition}

In the case where $\gamma=1$, the corresponding problem is simply called \SVP.
\begin{theorem}[Haviv and Regev~\cite{STOC:HavReg07}]
    \ASVP~is \NP-hard under randomized reductions for every constant
  approximation factor.
\end{theorem}

A variant of the problem consists of finding vectors within
Hermite-like inequalities.

\begin{definition}[$\gamma$-\HSVP]
  Let $\gamma = \gamma(n) \geq 1$. The $\gamma$-Hermite Shortest Vector
  Problem ($\gamma$-\HSVP) is defined as follows.
  \begin{description}
    \item[Input] A basis $(v_1, \ldots, v_n)$ of a lattice $\Lat$.
    \item[Output] A lattice vector $v\in\Lat\setminus\{0\}$ satisfying
      $\|v\|\leq \gamma \covol(\Lat)^{\frac1n}$.
  \end{description}
\end{definition}

There exists a simple polynomial-time dimension-preserving reduction between
these two problems, as stated by Lov\'asz in~\cite[1.2.20]{lovasz1986algorithmic}:
\begin{theorem}
  \label{thm:lovasz_reduction}
  One can solve $\gamma^2$-\SVP{} using $2n$ calls to a
  $\gamma$-\HSVP{} oracle and polynomial time.
\end{theorem}
This can be slightly improved in case the \HSVP{} oracle is built from a
\HSVP{} oracle in lower dimension~\cite{aggarwal2019slide}.

\subsubsection{An oracle for $\gamma$-\HSVP}

We note $T(\beta)$ a function such that we can solve
$\bigO{\sqrt{\beta}}$-\HSVP{} in time at most $T(\beta)$ times the input size.
We have the following bounds on $T$, depending on if we are looking at an
algorithm which is:
\begin{description}
  \item[Deterministic] $T(\beta)=(4+\littleO{1})^{\beta/2}$, proven by
Micciancio and Voulgaris~in\cite{SODA:MicVou10};
  \item[Randomized] $T(\beta)=(4/3+\littleO{1})^{\beta/2}$~, introduced
    by Wei, Liu and Wang in \cite{RSA:WeiLiuWan15};
  \item[Heuristic] $T(\beta) = (3/2+\littleO{1})^{\beta/2}$~in \cite{SODA:BDGL16} by Becker,
    Ducas, Gama, Laarhoven.
\end{description}
There also exists variants for quantum
computers~\cite{laarhoven2015finding}, and time-memory tradeoffs, such
as~\cite{PKC:HerKirLaa18}.
By providing a back-and-forth strategy coupled with enumeration in the
dual lattice, the \emph{self dual block Korkine-Zolotarev} (\DBKZ) algorithm
provides an algorithm better than the famous
\textsc{bkz} algorithm.

\begin{theorem}[Micciancio and Walter~\cite{EC:MicWal16}] \label{thm:BKZ}
  There exists an algorithm ouputting a vector $v$ of a lattice $\Lat$
  satisfying:
  \[\|v\| \leq \beta^{\frac{n-1}{2(\beta-1)}} \cdot
  \covol(\Lat)^{\frac{1}{n}}.\]
  Such a bound can be achieved in time
  $(n+\log \|B\|)^{\bigO{1}}T(\beta)$,
  where $B$ is the integer input basis representing $\Lat$.
\end{theorem}
\begin{proof}
  The bound we get is a direct consequence of~\cite[Theorem~1]{EC:MicWal16}.
  We only replaced the \emph{Hermite constant} $\gamma_\beta$ by an upper
  bound in $\bigO{\beta}$.
\end{proof}

A stronger variant of this estimate is heuristically true, at least for
``random'' lattices, as it is suggested by the Gaussian Heuristic
in~\cite[Corollary 2]{EC:MicWal16}. Under this assumption, one can bound
not only the length of the first vector but also the gap between the
covolumes of the filtration induced by the outputted basis.

\begin{theorem} \label{thm:HBKZ}
  There exists an algorithm ouputting a complete filtration of a lattice
	$\Lat$ satisfying: \[ \covol(\faktor{\Lambda_i}{\Lambda_{i-1}}) \approx
  \Theta(\beta)^{\frac{n+1-2i}{2(\beta-1)}} \covol(\Lat)^{\frac{1}{n}} \] Such
  a bound can be achieved in time $(n+\log \|B\|)^{\bigO{1}}T(\beta)$, where $B$
  is the integer-valued input basis.
  Further, we have:
  \[ \Theta(\sqrt{\beta})
    \covol^\frac1\beta{\left(\faktor{\Lambda_n}{\Lambda_{n-\beta}}\right)}
\approx \covol{\left(\faktor{\Lambda_{n-\beta+1}}{\Lambda_{n-\beta}}\right)}. \]
\end{theorem}

\subsection{The closest vector problem}
\label{sec:closest_vector_problem}
In this section we introduce formally the \CVP~problem and its
variants and discuss their computational hardness.

\begin{definition}[$\gamma$-\CVP]
  Let $\gamma = \gamma(n) \geq 1$. The $\gamma$-Closest Vector
  Problem ($\gamma$-\CVP) is defined as follows.

  \begin{description}
    \item[Input] A basis $(v_1, \ldots, v_n)$ of a lattice $\Lat$ and
      a target vector $t \in\Lat\otimes\RR$.
    \item[Output] A lattice vector $v\in\Lat$ satisfying
      $\|x-t\|\leq \gamma \min_{v\in\Lat} \|v-t\|$.
  \end{description}
\end{definition}

In the case where $\gamma=1$, the corresponding problem is called \CVP.

\begin{theorem}[Dinur, Kindler and Shafra~\cite{dinur1998approximating}]
  $n^{\frac{c}{\log \log n}}$-\ACVP~is \NP-hard for any $c>0$.
\end{theorem}

We let $T_{\CVP}(\beta)$ be such that we can solve \CVP{} in dimension $\beta$
in running time bounded by $T_{\CVP}(\beta)$ times the size of the input.
Hanrot and Stehl\'e proved $\beta^{\beta/2+\littleO{\beta}}$ with polynomial
memory~\cite{C:HanSte07}.
Sieves can provably reach $(2+\littleO{1})^{\beta}$ with exponential
memory~\cite{FOCS:AggDadSte15}.
More importantly for this paper, heuristic sieves can reach
$(4/3+\littleO{1})^{\beta/2}$ for solving an entire batch of $2^{0.058\beta}$
instances~\cite{sieve20}.

\section{The nearest colattice algorithm}
\label{sec:nearest_coplane}

We aim
at solving the $\gamma-$\ACVP{} by recursively exploiting
the datum of a filtration \[ \Lat_0 \subset \Lat_1 \subset \cdots \subset
\Lat_k = \Lat \]
\emph{via} recursive approximations. The central object used during
this reduction is the \emph{nearest colattice} relative to
a target vector.

In this section, and the next one, we assume that the size of the bases
is always small, essentially as small as the input basis.
This is classic, and can be easily proven.

\subsection{Nearest colattice to a vector}

\begin{definition}
  \label{def:nearest_colattice}
  Let $0\rightarrow
  \Lat'\rightarrow\Lat\rightarrow\faktor{\Lat}{\Lat'}\rightarrow 0$ be a
  short exact sequence of lattices, and
  set $t\in \Lat_\RR$ a target vector. A nearest
  $\Lat'$-colattice to $t$ is a coset $\bar{v} = v+\Lat'\in
  \faktor{\Lat}{\Lat'}$ which is the closest to the projection of $t$ in
  $\faktor{\Lat_\RR}{\Lat'_\RR}$, i.e. such that:
  \[
    \bar{v} = \argmin_{v\in\Lat}
    \|(t-v)+\Lat'\|_{{\Lat_\RR}/{\Lat'_\RR}}
  \]
\end{definition}

This definition makes sense thanks to the discreteness of
the quotient lattice $\faktor{\Lat}{\Lat'}$ in the real vector
space $\faktor{\Lat_\RR}{\Lat'_\RR}$.

\begin{exemple}
  To illustrate this definition, we give two examples in dimension 3, of rank 1
  and 2 nearest colattices.
  Set $\Lat$ a rank 3 lattice, and fix $\Lat_1$ and $\Lat_2$ two pure
  sublattices of respective rank 1 and 2. Denote by $\pi_i$ the canonical
  projection onto the quotient $\faktor{\Lat}{\Lat_i}$, which is of dimension
  $3-i$  for $i\in\{1,2\}$. The $\Lat_i$-closest colattice to $t$, denoted
  by $v_i+\Lat_i$ is such that $\pi_i(v_i)$ is a closest vector to $\pi_i(t)$
  in the corresponding quotient lattice. Figures (\textsc{a}) and (\textsc{b})
  respectively depict these situations.

  \begin{figure}\makebox[\textwidth][c]{
    \subfloat[The $\Lat_2$-nearest colattice $\color{duckBlue}v+\Lat_2$
    relative to $\color{goldOrange}t$, in green.]{
      \begin{tikzpicture}[scale=1.5]
        \nearestplane{}
      \end{tikzpicture}
    }\subfloat[The $\Lat_1$-nearest colattice $\color{duckBlue}v+\Lat_1$
    relative to $\color{goldOrange}t$.]{
      \begin{tikzpicture}[scale=1.5]
        \nearestcospace{}
      \end{tikzpicture}
    }}
  \end{figure}

\end{exemple}

\begin{remark}
  A computational insight on \cref{def:nearest_colattice} is
  to view a nearest colattice as a solution to an instance of
  exact-\CVP~in the quotient lattice $\faktor{\Lat}{\Lat'}$.
\end{remark}

Taking the same notations as in \cref{def:nearest_colattice}, let us
project $t$ orthogonally onto the affine space $v+\Lat'_\RR$, and take
$w$ a closest vector to this projection.
The vector $w$ is then relatively close to $t$. Let us quantify its defect
of closeness towards an actual closest vector to $t$:

\begin{proposition}
  \label{prop:defect_target}
  With the same notations as above:
  \[ \|t-w\|^2 \leq \mu\left(\faktor{\Lat}{\Lat'}\right)^2 +
  \mu\left(\Lat'\right)^2 \]
\end{proposition}
\begin{proof}
  Clear by Pythagoras' theorem.
\end{proof}

By definition of the covering radius, we then have:
\begin{corollary}[Subadditivity of the covering radius over short exact
  sequences]
  \label{cor:subadditivity}
  short exact sequence of lattices. Then we have:
  \[ \mu(\Lat)^2 \leq \mu\left(\faktor{\Lat}{\Lat'}\right)^2 +
  \mu\left(\Lat'\right)^2 \]
\end{corollary}

This inequality is tight, as being an equality when
there exists a sublattice $\Lat''$ such that $\Lat'\oplus \Lat'' = \Lat$ and
$\Lat'' \subseteq \orth{\Lat'}$.

\subsection{Recursion along a filtration}
Let us now consider a filtration \[ \Lat_0 \subset \Lat_1 \subset \cdots
  \subset \Lat_k = \Lat \] and a target vector $t\in\Lat_\RR$.
  Repeatedly applying \cref{cor:subadditivity} along the subfiltrations
  $0\subset \Lat_i \subset \Lat_{i+1}$,
  yields a sequence of inequalities $\mu(\Lat_{i+1})^2-\mu\left(\Lat_i\right)^2 \leq
  \mu\left({\Lat_{i+1}}/{\Lat_{i}}\right)^2$. The telescoping sum now
  gives the relation:
  \[
    \mu(\Lat)^2 \leq \sum_{i=1}^k
    \mu\left(\faktor{\Lat_{i+1}}{\Lat_i}\right)^2.
  \]
  This formula has a very natural algorithmic interpretation as a recursive
  oracle for approx-\CVP:
  \begin{enumerate}
    \item Starting from the target vector $t$, we solve the \CVP{} instance
      corresponding to $\pi(t)$ in the quotient $\faktor{\Lat_k}{\Lat_{k-1}}$
      with $\pi$ the canonical projection onto this quotient to find
      $v+\Lat_{k-1}$ the
      nearest $\Lat_{k-1}$-colattice to $t$.
    \item We then project $t$ orthogonally onto
      $v+({\Lat_{k-1}}\otimes_\ZZ\RR)$. Call $t'$
      this vector.
    \item A recursive call to the algorithm on the instance $(t'-v,
      \Lat_0\subset \cdots \subset \Lat_{k-1}))$
      yields a vector $w \in \Lat_2$.
    \item Return $w+v$.
  \end{enumerate}

Its translation in pseudo-code is given in an iterative manner in
the algorithm \algName{Nearest-Colattice}.
\begin{algo}[algotitle={Nearest-Colattice},
  label=alg:nearest_colattice]
  \sffamily
  \begin{algorithm}[H]
\BlankLine
      \KwIn{A filtration $\{0\} = \Lat_0 \subset \Lat_1 \subset \cdots
  \subset \Lat_k = \Lat$, a target $t\in\Lat_\RR$ }
      \KwResult{A vector in $\Lat$ close to $t$.}
\BlankLine
    $s \gets -t$\;
    \For{$i=k$ \Downto $1$}{$s \gets s-\algName{Lift}(\argmin_{h\in\Lambda_i/\Lambda_{i-1}}
        \|v-h\|)$\;
    }
    \Return $t+s$
  \end{algorithm}
\end{algo}

\begin{proposition}
  \label{prop:sum_covering}
  Let $B$ be a basis of a lattice $\Lat$ of rank $n$.
  Given a target $t\in\Lat_\RR$,
  the algorithm \algName{Nearest-Colattice} finds a vector $x\in \Lat$ such that
  \[ \|x-t\|^2 \leq \sum_{i=1}^k \mu\left(\faktor{\Lat_{i+1}}{\Lat_i}\right)^2
  \] in time $T_{\CVP}(\beta)(n+\log \|t\|+\log \|B\|)^{\bigO{1}}$, where $\beta$
  is the largest gap of rank in the filtration: $\beta = \max_i
  (\rk(\Lat_{i+1})-\rk(\Lat_{i}))$.
\end{proposition}
\begin{proof}
  The bound on the quality of the approximation is a direct consequence of the
  discussion conducted before. The running time bound derives from the
  definition of $T_{\mathrm{CVP}}$ and on the fact that the \algName{Lift} operations
  can be conducted in polynomial time.
\end{proof}

\begin{remark}[Retrieving Babai's algorithm]
  \label{rem:babai}
  In the specific case where the filtration is complete, that is
  to say that $\rk(\Lat_i)=i$ for each $1\leq i \leq n$, the
  \algName{Nearest-Colattice} algorithm coincides with the
  so-called \emph{Babai's nearest plane} algorithm. In particular, it recovers
  a vector at distance \[\sqrt{\sum_{i=1}^{n}
  \mu\left(\faktor{\Lat_i}{\Lat_{i-1}}\right)^2} =\frac12\sqrt{\sum_{i=1}^{n}
  \covol\left(\faktor{\Lat_i}{\Lat_{i-1}}\right)^2},\]
    by using that for each index $i$, we have
    $\mu\left(\faktor{\Lat_i}{\Lat_{i-1}}\right) =
    \frac12 \covol\left(\faktor{\Lat_i}{\Lat_{i-1}}\right)$ as these quotients are
      one-dimensional.
\end{remark}

The bound given in \cref{prop:sum_covering} is not easily instantiable
as it requires to have access to the covering radius of the successive
quotients of the filtration. However, under a mild heuristic on random
lattices, we now exhibit a bound which only depends on the parameter $\beta$ and
the covolume of $\Lat$.

\subsection{On the covering radius of a random lattice}

In this section we prove that the covering radius of a random lattice
behaves essentially in $\sqrt{\rk(\Lambda)}$.

In 1945, Siegel~\cite{Siegel} proved that the projection of the Haar measure
of $\SL_n(\RR)$
over the quotient $\SL_n(\RR)/\SL_n(\ZZ)$ is of finite mass, yielding a
natural probability
distribution $\nu_n$ over the moduli space $\mathcal{L}_n$ of unit-volume
lattices. By construction this distribution is translation-invariant,
that is, for any measurable set $\mathcal{S} \subseteq \mathcal{L}_n$ and all
$U \in \SL_n(\ZZ)$, we have $\nu_n(\mathcal{S}) = \nu_n(\mathcal{S}U)$. A
\emph{random lattice} is then defined as a unit-covolume lattice in $\RR^n$
drawn under the probability distribution $\nu_n$.

We first recall an estimate due to Rogers~\cite{rogers1955mean}, giving the
expectation\footnote{The result proved by Rogers is actually more general and
bounds all the moment of the enumerator of lattice points. For the purpose of
this work, only the first moment is actually required.} of the number of
lattice points in a fixed set.

\begin{theorem}[Rogers' average]\label{thm:rogers}
  Let $n\leq 4$ be an integer
  and $\rho$ be the characteristic function of a Borel set $C$ of $\RR^n$
  whose volume is $V$, centered at 0. Then:
  \[
    \begin{aligned}
      0  \le  \int_{\mathcal{L}_n} \rho(\Lat\setminus \{0\})  d\nu_n(\Lat)
      &-
      2
      e^{-V/2} \sum_{r=0}^{\infty} \frac{r}{r!} (V/2)^r \\
      &
      \le \left(V+1\right)
      \left(
        6 \left(\sqrt{\frac{3}{4}}\right)^n + 105\cdot 2^{-n}
      \right).
    \end{aligned}
    \]
\end{theorem}

This allows to prove that the first minimum of a random lattice is
greater than a multiple of $\sqrt{n}$.

\begin{lemma}
  \label{lem:estimate_lambda1}
  Let $\Lambda$ be a random lattice of rank $n$. Then, with probability
  $1-2^{-\Omega(n)}$, $\lambda_1(\Lambda)>c\sqrt{n}$ for a universal constant
  $c>0$.
\end{lemma}
\begin{proof}
  Consider the ball $C$ of volume $0.99^n$.
  It has a radius lower bounded by $c\sqrt{n}$.
  By \cref{thm:rogers}, the expectation of
  the number of lattice points in $C$ is at most
  \[128\left(\frac34\right)^{\frac{n}{2}}(V+1)+ V \in (1+\littleO{1})V. \]
  This estimate thus bounds the probability that there exists a non-zero
  lattice vector in $C$ by $1-2^{-\Omega(n)}$, using Markov's inequality.
\end{proof}

Using the transference theorem, we then derive the following estimate
on the covering radius of a random lattice:
\begin{theorem}
  Let $\Lambda$ be a random lattice of rank $n$.
  Then, with probability $1-2^{-\Omega(n)}$, $\mu(\Lambda)<d\sqrt{n}$ for a
  universal constant $d$.
\end{theorem}
\begin{proof}
  First remark that the dual lattice $\dual{\Lambda}$ follows the same
  distribution.  Hence, using the estimate of \cref{lem:estimate_lambda1}, we
  know that with probability $1-2^{-\Omega(n)}$,
  $\lambda_1(\dual{\Lambda})>c\sqrt{n}$. Banaszczyk's transference theorem
  indicates that in this case, \[\mu(\Lambda)\leq
  \frac{n}{\lambda_1(\dual{\Lambda})} \leq \frac{\sqrt{n}}{c},\]
  concluding the proof.
\end{proof}

This justifies the following heuristic:
\begin{heuristic}
  \label{heur:lambda1_heur}
  In algorithm \algName{Nearest-Colattice}, for any index $i$, we have
  $\mu\left(\faktor{\Lat_{i+1}}{\Lat_i}\right)\leq c
  \lambda_1\left(\faktor{\Lat_{i+1}}{\Lat_i}\right)$ for some
  universal constant $c$.
\end{heuristic}
The Gaussian heuristic suggests that ``almost all'' targets $t$ are at
distance $(1+\littleO{1})\lambda_1(\Lambda)$, so that for practical purpose in
the analysis we can take $c=1$ in \cref{heur:lambda1_heur}.

\subsection{Quality of the algorithm on random lattices}

\begin{theorem}
  \label{thm:quality_complexity_average}
  Let $\beta>0$ be a positive integer and $B$ be a basis of a lattice $\Lat$
  of rank $n>2\beta$.
  After precomputations using a time bounded by
  $T(\beta)(n+\log \|B\|)^{\bigO{1}}$, given a target $t\in\Lat_\RR$ and under
  \cref{heur:lambda1_heur}, the algorithm \algName{Nearest-Colattice}
  finds a vector $x\in \Lat$ such that \[ \|x-t\| \leq
  \Theta(\beta)^{\frac{n}{2\beta}}\covol(\Lat)^{\frac{1}{n}} \] in time
  $T_{\CVP}(\beta)\textrm{Poly}(n,\log \|t\|,\log \|B\|)$.
\end{theorem}
\begin{proof}
  We start by reducing the basis $B$ of $\Lat$ using the \DBKZ{} algorithm,
  and collect the vectors in blocks of size $\beta$, giving a filtration:
  \[
    \{0\} = \Lat_0 \subset \Lat_1 \subset \cdots \subset \Lat_k = \Lat,
  \]
  for $k = \left\lceil\frac{n}{\beta}\right\rceil$ and
  $\rk\left(\faktor{\Lat_{i+1}}{\Lat_i}\right) = \beta$ for each index $i$
  except the penultimate one, of rank
  $n-\beta\left\lfloor\frac{n}{\beta}\right\rfloor$.
  We define $l_i$ as $\rk\left(\faktor{\Lat_{i+1}}{\Lat_i}\right)$.
  By \cref{thm:HBKZ} and
  finite induction in each block using the multiplicativity of the covolume
  over short exact sequences, we have for $i<k-1$
  \[
    \begin{aligned}
      \covol\left(\faktor{\Lat_{i+1}}{\Lat_i}\right)^{\frac{1}{l_i}}& \approx
\covol\left(\Lat\right)^\frac1n
      \left(\prod_{j=i\beta}^{i\beta+l_i-1}
      \Theta(\beta)^{\frac{n+1-2j}{2(\beta-1)}}\right)^{\frac{1}{l_i}}
       \\
                                                    &=
                                                    \Theta(\beta)^{\frac{
                                                        n + 2 - 2i\beta-l_i
                                                    }{2(\beta-1)}}
                                                    \covol(\Lat)^\frac{1}{n}.
    \end{aligned}
  \]
  We also have
  \[
  \Theta(\sqrt{\beta})\covol\left(\faktor{\Lat_{k}}{\Lat_{k-1}}\right)^{1/\beta}\approx
  \Theta(\beta)^{\frac{n+1-2(n-\beta)}{2(\beta-1)}}\covol^{\frac1n} \Lat
  \]
  so that the previous approximation is also true for $i=k-1$.
  Using \cref{heur:lambda1_heur} and Minkowski's first theorem, we can
  estimate the covering radius of this quotient as:
  \[
    \mu\left(\faktor{\Lat_{i+1}}{\Lat_i}\right) \leq \Theta(\sqrt{l_i})
  \Theta(\beta)^{\frac{n+2-2i\beta-l_i}{2(\beta-1)}}\covol^\frac1n{\Lat}.\]
Using \cref{prop:sum_covering}, now asserts that \algName{Nearest-Colattice}
returns a vector at distance from $t$ bounded by:
\[
  \covol(\Lat)^\frac1n\sum_{i=0}^k \Theta(\sqrt{l_i})
\Theta(\beta)^{\frac{n+2-2i\beta-l_i}{2(\beta-1)}}=
  \Theta(\beta)^{\frac{n}{2\beta-2}}\covol(\Lat)^\frac1n
\]
where the last equality stems from the condition $n\geq 2\beta$, so that only
  the first term is significant.
\end{proof}

Note that in the algorithm, all lattices depend only on $\Lat$, not on the
targets.  Therefore, it is possible to use \CVP{} algorithms after
precomputations.  These algorithms are significantly faster; we refer
to ~\cite{sieve20} for heuristic ones and
to~\cite{dadush2014closest,stephens2019time} for proven approximation
algorithms.
 \section{Proven \ACVP{} algorithm with precomputation}
In all of this section, let us fix an oracle
\oracle, solving the $\gamma$-\HSVP.
We solve \ACVP{}
with preprocessing from the oracle $\mathcal{O}$.

\begin{theorem}[\ACVPP{} oracle from \HSVP{} oracle]
  \label{thm:main_result}
  Let $\Lat$ be a lattice of rank $n$. Then one can solve the
  $(n^\frac32\gamma^3)$-closest vector problem in $\Lat$, using $2n^2$ calls
  to the oracle \oracle{} during precomputation, and polynomial time
  computations.
\end{theorem}

The first step of this reduction consists in proving that we can find a
lattice point at a distance roughly $\lambda_n(\Lambda)$.

\begin{theorem}\label{thm:absolu}
  Let $\Lat$ be a lattice of rank $n$ and $t\in\Lat\otimes\RR$ a target
  vector, then  one can find a lattice vector $c\in\Lat$ satisfying
  \[\|c-t\|\leq\frac{\sqrt{n}\gamma}{2}\lambda_n(\Lambda),\]
  using $n$ calls to the oracle \oracle{} during precomputation, and
  polynomial time computations.
\end{theorem}
\begin{proof}
  We aim at constructing a complete filtration \[ \{0\} \subset \Lat_1 \subset
  \cdots \subset \Lat_n = \Lat\] of the input lattice $\Lat$ such that for any
  index $1\leq i\leq n-1$, we have: \[
  \covol\left(\faktor{\Lat_i}{\Lat_{i-1}}\right) \leq \gamma
\lambda_n(\Lat).\]
We proceed inductively:
\begin{itemize}
  \item By a call to the oracle \oracle{} on the lattice $\Lat$, we find a
    vector $b_1$. Set $\Lat_1 = b_1\ZZ$ the corresponding sublattice.
  \item Suppose that the filtration is constructed up to index $i$.
    Then we call the oracle \oracle{} on the quotient sublattice
    $\faktor{\Lat}{\Lat_i}$ (or equivalently on the projection of $\Lat$
    orthogonally to $\Lat_i$), and lift the returned vector
    using the \algName{lift} function in $v \in \Lat$. Eventually we set
    $\Lat_{i+1} = \Lat_i\oplus v\ZZ$.
\end{itemize}

  At each index, we have by construction
  $\lambda_{n-i+1}\left(\faktor{\Lat}{\Lambda_{i}}\right)\leq
  \lambda_n(\Lambda)$. As such,  $\covol\left(\faktor{\Lat}{\Lambda_{i}}\right)\leq
  \lambda_n(\Lambda)^{n-i+1}$, and, eventually, we have for each index $i$:
  \[ \covol\left(\faktor{\Lat_i}{\Lat_{i-1}}\right)\leq
  \gamma\cdot\lambda_n(\Lambda). \]

  As stated in \cref{rem:babai}, Babai's algorithm on the point $t$ returns a
  lattice vector $c\in\Lat$ such that:
  \[ \|c-t\|\leq \sqrt{\sum_{i=1}^{n}
    \mu\left(\faktor{\Lat_i}{\Lat_{i-1}}\right)^2}\leq
    \frac{\sqrt{n}\gamma\lambda_n(\Lambda)}{2}.\]
\end{proof}

\begin{remark}[On the quality of this decoding]
  For a random lattice, we expect $\lambda_n(\Lambda)\approx
  \sqrt{n}\covol(\Lambda)^{\frac1n}$, so that the distance between the decoded
  vector and the target is only a factor $\gamma$ times larger than the guaranteed
  output of the oracle.
\end{remark}

We can now complete the reduction:\\
\noindent\textit{Proof of \cref{thm:main_result}.}
Let $\Lat$ be a rank $n$ lattice.
Without loss of generality, we might assume that the norm $\|.\|$ of $\Lat$
coincides with its dual norm, so that the dual $\dual{\Lat}$ can be
isometrically embedded in $\Lat_\RR$.
  We first find a non-zero vector in the dual lattice: $c\in \dual{\Lambda}$, where
  $\|c\|\leq \gamma^2
  \lambda_1(\dual{\Lambda})$ using Lov\'asz's reduction stated in
  \cref{thm:lovasz_reduction} on the oracle \oracle.
  Define $v\in\Lat$ and $e\in\Lat\otimes\RR$ to satisfy
  $t=v+e$ with $\|e\|$ minimal. We now have two cases, depending on
  how large is the error term $e$:
  \begin{description}
    \item[Case $\|c\|\|e\|\geq 1/2$ (large case)] Then, by pluging
      Banaszczyk's transference inequality to the bound on $\|c\|$ we get:
      \[ \|e\|\geq \frac{1}{2\gamma^2\lambda_1(\dual{\Lambda})}\geq
      \frac{\lambda_n(\Lambda)}{2n\gamma^2}. \]
      Thus, we can use \cref{thm:absolu} to solve \ACVP{} with
      approximation factor equal to:
      \[
        \frac{\sqrt{n}
        \gamma}{2}\left({\frac{1}{2n\gamma^2}}\right)^{-1}=n^{\frac32}\gamma^3.
      \]

    \item[Case $\|c\|\|e\|<1/2$ (small case)]
  \end{description}
      \parbox[t]{\dimexpr\textwidth-\leftmargin+4em}{Then, we have by linearity
      $\inner{c}{t}=\inner{c}{v}+\inner{c}{e}$. Hence, by the Cauchy-Schwarz
      inequality and the assumption on $\|c\|\|e\|$ we can assert that:
      \[ \lfloor \inner{c}{t} \rceil=\inner{c}{v}. \]
      Let $\Lat'$ be the projection of $\Lambda$ over the orthogonal space to
      $c$ and denote by $\pi$ the corresponding orthogonal projection.
      \begin{wrapfigure}{L}{0.5\textwidth}
        \centering
        \begin{tikzpicture}[scale=2.5]
\draw [lightgray] [->] (0,0) -- (0, 1.85);
          \draw [lightgray] [->] (0,0) -- (2.3,0);
          \node[lightgray] at (2.5,0) {\small $c\RR$};
          \node[lightgray] at (0,2) {\small $(c\RR)^\bot$};

          \draw [goldOrange, dashed] [] (1.5,0.5) -- (0,0.5);
          \draw [duckBlue, dashed ] [] (1.2,1.5) -- (0,1.5);
          \draw [goldOrange] []  (-0.25,0.5) node {\small $\pi(t)$};
          \draw [duckBlue] []  (-0.25,1.5) node {\small $\pi(v)$};
          \draw [fill,lightgray] (0,0.5) circle [radius=0.03];
          \draw [fill,lightgray] (0,1.5) circle [radius=0.03];
\draw [fill, goldOrange, opacity=0.3] (1,0) -- (1,2) -- (2,2) -- (2,0);
          \node[goldOrange] at (1.9,1.9) {\small $\mathcal{D}$};
          \draw [goldOrange, opacity=1] [] (1,0) -- (1,2);
          \draw [goldOrange, opacity=1] [] (2,2) -- (2,0);
          \draw [thick, goldOrange, shorten >=0.5] [->] (0,0) -- (1.5,0.5);
          \draw [thick, duckBlue, shorten >=0.7, shorten <=0.8] [<-] (1.2,1.5) -- (1.5,0.5);
          \draw [fill,goldOrange,opacity=.4] (1.5,0.5) circle [radius=0.03];
          \draw [goldOrange] (1.6,0.4) node {\small $t$};

          \draw [fill,duckBlue, opacity=.8] (1.2,1.5) circle [radius=0.03];
          \draw [duckBlue] (1.1,1.65) node {\small $v$};
          \draw [duckBlue, dashed] [] (1.2,0) -- (1.2,1.5);
          \draw [goldOrange, dashed] [] (1.5,0) -- (1.5,0.5);
          \node[goldOrange, rotate=-90] at (1.5,-0.21) {\tiny $\inner{c}{t}$};
          \node[goldOrange, rotate=-90] at (1,-0.35) {\tiny $\inner{c}{t}-\frac12$};
          \node[goldOrange, rotate=-90] at (2,-0.35) {\tiny $\inner{c}{t}+\frac12$};
          \node[duckBlue, rotate=-90] at (1.2,-0.21) {\tiny $\inner{c}{v}$};
          \node[gray, rotate=-90] at (0.2,-0.35) {\tiny $\inner{c}{v}-1$};
          \node[gray, rotate=-90] at (2.2,-0.35) {\tiny $\inner{c}{v}+1$};

\draw [gray] []  (-0.25,1) node {\small $\tilde{p}$};
          \draw [fill,gray] (0,1) circle [radius=0.03];
          \draw [fill,gray] (1.2,1) circle [radius=0.03];
          \draw [fill,gray] (0.2,1) circle [radius=0.03];
          \draw [fill,gray] (2.2,1) circle [radius=0.03];
          \draw [gray, dashed] [] (0.2,0) -- (0.2,1);
          \draw [gray, dashed] [] (2.2,0) -- (2.2,1);
          \draw [gray] [] (0,1) -- (2.2,1);
          \node[gray] at (2.5,1) {\tiny $\pi^{-1}(\tilde{p})$};
        \end{tikzpicture}
        \caption{\label{fig:illustatration_proof}Illustration of the situation
        depicted in the proof, in the two dimensional case.}
      \end{wrapfigure}

      Let us prove that $\pi(v)$ is a closest vector of $\pi(t)$ in
      $\Lat'$. To do so, let us take $\tilde{p}$ a shortest vector $\pi(t)$ in
      $\Lat$. We now look at the fibre (in $\Lat$) above $\tilde{p}$ and take
      the closest element $p$ to $t$ in this set. Then by Pythagoras' theorem,
      $p$ is an element of the intersection of
      $\pi^{-1}(\tilde{p})$ with the convex body $\mathcal{D} = \left\{x \,|\,
  |\inner{c}{x}|<\frac12\right\}$. As the vector $c$ belongs to the dual
  of $\Lat$, we have that for any $p_1, p_2\in\pi^{-1}(\tilde{p}),
  \inner{p_1-p_2}{c}\in\ZZ$, so that $\pi^{-1}(\tilde{p})\cap\mathcal{D}$ is
  of cardinality one. Write $p$ for this point. Then, $\inner{p}{c} =
  \inner{v}{c}$, as $|\inner{p-v}{c}|<1/2$ and is an integer.
  Now remark that by
  minimality of $\|v-t\|$, we have by Pythagoras' theorem that $v=p$, implying
  that $\pi(v)=\tilde{p}$.\\

By induction, we find $w\in \Lat$ such that
\[ \|\pi(w-t)\|\leq n^{3/2}\gamma^3 \|\pi(v-t)\| \]
and since $\inner{c}{w-t}=\inner{c}{v-t}$ we obtain
\[ \|w-t\|\leq n^{3/2}\gamma^3 \|v-t\|. \]
  \qed
    }
Overall, we get the following corollary by using the Micciancio-Voulgaris
algorithm for the oracle \oracle:
\begin{corollary}
  We can solve $\beta^{\bigO{\frac{n}{\beta}}}$-\ACVP{} deterministically in
  time bounded by $2^\beta$ times the size of the input.
\end{corollary}

\begin{remark}
  Using exactly the scheme proof scheme, we can refine the approximation
  factor to a $n^{3/2}\gamma_S \gamma$ by using a separate $\gamma_S$-SVP
  oracle instead of using $\gamma$-\HSVP{} as a $\gamma^2$-\SVP{} oracle.
\end{remark}

\section{Cryptographic perspectives}
\label{sec:applications}

In cryptography, the \textsc{Bounded Distance Decoding} (\textsc{bdd})
problem\footnote{ This problem being defined as finding the closest lattice
vector of a target, provided it is within a fraction of $\lambda_1(\Lambda)$.}
has a lot of importance, as it directly relates to the celebrated
\emph{Learning With Error} problem (\textsc{lwe})~\cite{regev2009lattices}.
This latter problem can be reduced to \ACVP, however our theoretical reduction
with \HSVP{} has a loss which is too large to be competitive.

In the so-called GPV framework~\cite{STOC:GenPeiVai08}, instantiated in the
\texttt{DLP} cryptosystem~\cite{AC:DucLyuPre14} and its follow-ups
\texttt{Falcon}~\cite{fouque2018falcon}, \texttt{ModFalcon}~\cite{modfalcon},
a valid signature is a point close to a target, which is the hash of the
message. Hence, forging a signature boils down to finding a close vector to a
random target.  Our first (heuristic) result implies that, once a reduced
basis has been found, forging a message is relatively easy.
Previous methods such as in~\cite{fouque2018falcon} used Kannan's
embedding~\cite{kannan1987minkowski} so that the cost given only applies for
one forgery, whereas a batch forgery is possible for roughly the same cost.

The same remark applies for practically solving the \textsc{bdd} problem, and
indeed the \textsc{lwe} problem.  Once a highly reduced basis is found, it is
enough to compute a \CVP{} on the tail of the basis, and finish with Babai's
algorithm. More precisely, by using the same notations an exploiting the proof of
\cref{thm:quality_complexity_average}, a sufficient condition for decoding will
be:
\[
  \|\pi(e)\| \leq \theta(\beta)^{\frac{2\beta-n}{2\beta}}\covol(\Lat)^\frac1n,
\]
where, $\pi$ is the orthogonal projection onto $\faktor{\Lat}{\Lat_k}$ and
$\beta$ is the rank of this latter lattice.

This trick seems to have been in the folklore for some time, and
is the reason given by \texttt{NewHope}~\cite{USENIX:ADPS16} designers for
selecting a random ``$a$'', which corresponds to a random lattice (where the
authors of~\cite{USENIX:ADPS16} claim that Babai's algorithm is enough, but it
seems to be practically true in general for an extremely well reduced basis,
i.e.\ with more precomputations performed).

\section*{Acknowledgments}
This work was done while the authors were visiting the Simons Institute for
the theory of computing in February 2020. They also thanks the anonymous
reviewers for their insightful comments on this work.

\end{document}